\documentclass[english]{article}
\usepackage[T1]{fontenc}
\usepackage[latin9]{inputenc}
\usepackage{geometry}
\usepackage{float}
\usepackage{dsfont}
\usepackage{amsmath}
\usepackage{amsthm}
\usepackage{amssymb}

\makeatletter

\floatstyle{ruled}
\newfloat{algorithm}{tbp}{loa}
\providecommand{\algorithmname}{Algorithm}
\floatname{algorithm}{\protect\algorithmname}

\theoremstyle{plain}
\newtheorem{thm}{\protect\theoremname}[section]
\theoremstyle{plain}
\newtheorem{prop}{\protect\propositionname}[section]
\theoremstyle{plain}
\newtheorem{lem}{\protect\lemmaname}[section]
\theoremstyle{remark}
\newtheorem{claim}{\protect\claimname}[section]
\theoremstyle{plain}
\newtheorem{fact}{\protect\factname}[section]

\date{}
\usepackage{algorithmic}
\usepackage{hyperref}

\makeatother

\usepackage{babel}
\providecommand{\claimname}{Claim}
\providecommand{\factname}{Fact}
\providecommand{\lemmaname}{Lemma}
\providecommand{\propositionname}{Proposition}
\providecommand{\theoremname}{Theorem}

\begin{document}
\global\long\def\PP{\mathsf{PrivatePerceptron}}%
\global\long\def\PPE{\mathsf{PrivatePerceptronEpoch}}%
\global\long\def\R{\mathbb{R}}%
\global\long\def\Lap{\mathrm{Lap}}%
\global\long\def\avg{\mathrm{avg}}%
\global\long\def\E{\mathbb{E}}%
\global\long\def\F{{\cal F}}%
\global\long\def\P{{\cal P}}%
\global\long\def\B{{\cal B}}%
\global\long\def\one{\mathds{1}}%
\global\long\def\vol{\mathrm{vol}}%
\global\long\def\diag{\mathrm{diag}}%
\allowdisplaybreaks
\title{Solving Linear Programs with Differential Privacy}
\author{Alina Ene\thanks{Department of Computer Science, Boston University, \texttt{aene@bu.edu}.}\and 
Huy L. Nguyen\thanks{Khoury College of Computer Sciences, Northeastern University, \texttt{hu.nguyen@northeastern.edu}.}\and 
Ta Duy Nguyen\thanks{Department of Computer Science, Boston University, \texttt{taduy@bu.edu}.}\and 
Adrian Vladu\thanks{CNRS \& IRIF, Universit\'e Paris Cit\'e, \texttt{vladu@irif.fr}.}}
\maketitle
\begin{abstract}
We study the problem of solving linear programs of the form $Ax\le b$,
$x\ge0$ with differential privacy. For homogeneous LPs $Ax\ge0$,
we give an efficient $(\epsilon,\delta)$-differentially private algorithm
which with probability at least $1-\beta$ finds in polynomial time
a solution that satisfies all but $O(\frac{d^{2}}{\epsilon}\log^{2}\frac{d}{\delta\beta}\sqrt{\log\frac{1}{\rho_{0}}})$
constraints, for problems with margin $\rho_{0}>0$. This improves
the bound of $O(\frac{d^{5}}{\epsilon}\log^{1.5}\frac{1}{\rho_{0}}\mathrm{poly}\log(d,\frac{1}{\delta},\frac{1}{\beta}))$
by {[}Kaplan-Mansour-Moran-Stemmer-Tur, STOC '25{]}. For general LPs
$Ax\le b$, $x\ge0$ with potentially zero margin, we give an efficient
$(\epsilon,\delta)$-differentially private algorithm that w.h.p drops
$O(\frac{d^{4}}{\epsilon}\log^{2.5}\frac{d}{\delta}\sqrt{\log dU})$
constraints, where $U$ is an upper bound for the entries of $A$
and $b$ in absolute value. This improves the result by Kaplan et
al. by at least a factor of $d^{5}$. Our techniques build upon privatizing
a rescaling perceptron algorithm by {[}Hoberg-Rothvoss, IPCO '17{]}
and a more refined iterative procedure for identifying equality constraints
by Kaplan et al. 

\newpage{}
\end{abstract}

\section{Introduction}

Linear programming is a fundamental tool for modeling and solving
problems in computer science. Consider the standard feasibility problem
of finding $x\in\R^{d}$ subject to $Ax\le b,x\geq0$, where the constraints
$Ax\leq b$ are input from users. The input can contain sensitive
information such as the users' private health data or private transactions,
which the users may wish to be protected. For the algorithm designer,
this is where differential privacy proves its usefulness. Differential
privacy protects sensitive information by requiring that the algorithm
must have approximately the same output upon receiving similar input.

The study of solving linear programs with differential privacy was
initiated by \cite{hsu2014privately} and has subsequently been studied
under different contexts along with other related problems. Notably,
privately solving linear programs is closely related to private learning
of subspaces and halfspaces, which are fundamental problems in learning
theory. In particular, a line of work by \cite{DBLP:conf/focs/BunNSV15,DBLP:conf/colt/BeimelMNS19,kaplan2020private,DBLP:conf/citc/GaoS21,DBLP:conf/nips/Ben-EliezerMZ22}
showed reductions from learning halfspaces  to the problem of finding
a point in a convex hull, which in turn can be solved via linear programing.
This means that the existence of an efficient private solver for linear
programs implies an upper bound for the sample complexity of efficient
private learners for halfspaces, and any improvement for the former
problem implies an improvement for the latter. 

Imposing differential privacy when solving a linear program comes
with the impossibility to ensure \emph{all} constraints are satisfied.
Indeed, in the extreme case, the addition of one more constraint can
change the problem from feasible to infeasible. Therefore, to guarantee
differential privacy, it is required that a number of constraints
must be dropped. It was known (by folklore) that for a linear program
whose feasible region is of positive volume, privatizing the algorithm
by \cite{dunagan2004polynomial} results in a solution that violates
$\mathrm{poly}(d)$ constraints, where $d$ is the dimension of the
problem. The recent work by \cite{kaplan2024differentially} formalized
this claim and generalized it to all linear programs. Their algorithm,
however, suffered from a high degree polynomial dependence on $d$
(at least $d^{9}$ dependence), whereas the lower bound (from learning
theory) is linear in $d$. Closing this gap remains a challenging
open question. 

In our work, we make progress in this direction, and propose new algorithms
for solving linear programs with differential privacy. Our algorithms
are efficient and they achieve significantly improved guarantees on
the number of violated constraints.

\subsection{Our contribution\protect\label{subsec:Our-contribution}}

Our first contribution is a new algorithm for privately solving linear
programs with positive margin with guarantee stated in Theorem \ref{thm:perceptron-guarantee}.
Here the margin of the problem is the radius of the largest ball that
fits in the intersection of the feasible region and the unit ball,
which we define in Section \ref{sec:Preliminary}.
\begin{thm}
\label{thm:perceptron-guarantee}Let $\epsilon,\delta,\beta\geq0$
be given an input. There exists an efficient $(\epsilon,\delta)$-differentially
private algorithm for finding a feasible solution to the linear program
$Ax\ge0$, $x\neq0$ such that whenever the margin of the system is
at least $\rho_{0}$, with probability at least $1-\beta$, the algorithm
outputs a solution $x$ that satisfies all but $O\left(\frac{d^{2}}{\epsilon}\sqrt{\log\frac{1}{\rho_{0}}}\log^{2}\frac{d}{\beta\delta}\right)$
constraints.
\end{thm}
In terms of the dependence on the dimension $d$, our algorithm significantly
improves over the prior work by \cite{kaplan2024differentially},
which drops $O\left(d^{5}\log^{1.5}\frac{1}{\rho_{0}}\mathrm{poly}\log\left(d,\frac{1}{\delta},\frac{1}{\beta}\right)\right)$
constraints. 

For general linear programs with potentially zero margin, we give
an iterative private algorithm with the following guarantee.

\begin{thm}
\label{thm:general-guarantee}Let $\epsilon,\delta,\beta\geq0$ be
given an input. There exists an efficient $(\epsilon,\delta)$-differentially
private algorithm for finding a feasible solution to the linear program
$Ax\le b$; $x\ge0$ with integer entries such that with probability
at least $1-(\beta+\delta)$, the algorithm outputs a solution $x$
that satisfies all but $O\left(\frac{d^{4}}{\epsilon}\sqrt{\log dU}\log^{2.5}\frac{d}{\beta\delta}\right)$
constraints, where $U$ is an upper bound on the absolute values of
the entries in $A$ and $b$.
\end{thm}
The algorithm by \cite{kaplan2024differentially} requires to drop
$\tilde{O}(d^{9})$ constraints to guarantee privacy. To improve this,
one can use the algorithm from Theorem \ref{thm:perceptron-guarantee}
as a subroutine in the algorithm by \cite{kaplan2024differentially}
to reduce to the number of dropped constraints to $\tilde{O}(d^{5})$.
Our algorithm with a more refined analysis goes one step further to
remove another factor $d$ and achieves $\tilde{O}(d^{4})$ dependence.
Our bound is a significant improvement towards the lower bound $\Omega(d)$.

\subsection{Our techniques\protect\label{subsec:Our-technique}}

Our technique for showing Theorem \ref{thm:perceptron-guarantee}
is based on privatizing a rescaling perceptron algorithm for solving
linear programs of the form $Ax\ge0,x\neq0$ with positive margin.
Instead of using the algorithm by \cite{dunagan2004polynomial} as
in \cite{kaplan2024differentially}, we develop an algorithm based
on the work of \cite{hoberg2017improved}. 

A rescaling perceptron algorithm consists of two phases: the Perceptron
Phase to find a new solution and the Rescaling Phase to find a direction
to rescale the input. To privatize the Perceptron Phase, the algorithm
of \cite{kaplan2024differentially} uses the $\mathsf{NoisyAvg}$
mechanism by \cite{nissim2016locating} to construct a noisy average
of the constraints that are violated by the current solution, and
then updates the solution in the same direction. In the Rescaling
Phase, the non-private algorithm by \cite{dunagan2004polynomial}
and the privatized version by \cite{kaplan2024differentially} use
another perceptron-style procedure to find a rescaling direction.
This approach requires $\tilde{O}(d^{2})$ updates and each call to
the Rescaling Phase requires to drop $\tilde{O}(d^{2.5})$ constraints.

The main difference between our work and \cite{kaplan2024differentially}
lies in the rescaling phase. We use the following technique by \cite{soheili2013deterministic}
and \cite{hoberg2017improved}. During the perceptron phase, the algorithm
maintains a weight vector $\lambda$ for tracking which constraints
are violated by the solution in each iteration. Using a random Gaussian
vector, the algorithm produces a rescaling direction by a convex combination
weighted by $\lambda$ of the rows satisfied by the random vector.
\cite{soheili2013deterministic} and \cite{hoberg2017improved} show
that with a constant probability, rescaling the input matrix along
that direction increases the volume of the feasible region inside
the unit ball by a constant factor. 

The benefit of using the weight vector to rescale is that, since we
only need $O(\log\frac{1}{\beta}$) steps to boost the success probability,
the amount of noise needed to make this phase private is much smaller.
Further, we can keep all constraints around until we find a good solution
and only discard constraints once at the end of the algorithm. We
discard only $O(d^{2})$ constraints in total.

For general linear programs with potentially zero margin, we use the
standard technique of adding a small perturbation to the constraints
that does not change the feasibility of the problem. This perturbation
increases the problem margin and allows the application of the private
perceptron algorithm. Similar to \cite{kaplan2024differentially},
our algorithm iteratively identifies tight (equality) constraints
in the LP, privatizes these equality constraints and uses them to
eliminate variables. The approach by \cite{kaplan2024differentially}
needs to account for the blow-up of the input entries after performing
the variable elimination steps, which can quickly diminish the margin.
The cost of this shows up as an extra factor $d$ in the number of
dropped constraints. By contrast, our algorithm always returns to
the initial LP after identifying tight constraints. We show that in
this way the margin reduces at a slower rate, saving the factor $d$.

\subsection{Related work}

\textbf{Rescaling Perceptron algorithms}. To find a solution $x\neq0$
that satisfies $Ax\ge0$, one can use the classic perceptron algorithm
which convergences after at most $1/\rho^{2}$ iterations on problems
with positive margin $\rho$, where the margin is the radius of the
largest ball that fits in the intersection of the feasible region
and the unit ball. \cite{dunagan2004polynomial} show a modification
of the classic algorithm with an additional rescaling procedure that
runs in time $\tilde{O}(nd^{4}\log\frac{1}{\rho})$, for problems
with $\rho>0$. The rescaling procedure by \cite{dunagan2004polynomial}
is another variant of the perceptron algorithm which finds a rescaling
direction by moving a random unit vector along the direction of a
violated constraint. Subsequent works by \cite{soheili2013deterministic,hoberg2017improved}
explore different rescaling operations. In particular, our work relies
on the technique by \cite{hoberg2017improved} in which the rescaling
direction is found by a convex combination of rows whose corresponding
constraints are satisfied by a random Gaussian vector.

\textbf{Solving linear programs with privacy}. Solving linear programs
with differential privacy has been the focus of several prior works.
Under various notions of neighboring inputs (such as input differing
by a row or a column), \cite{hsu2014privately} give algorithms that
approximately satisfy most constraints. \cite{munoz2021private} show
an algorithm that satisfy most constraints exactly, but only considering
neighboring inputs to be differing on the right hand side scalars.
\cite{kaplan2020private} provide an algorithm to solve the general
feasibility problem $Ax\ge b$, but requires a running time that is
exponential in the dimension. Most relevant for our work is the work
of \cite{kaplan2024differentially}, which studies the same setting.
We provide a detailed comparison of the results and techniques in
sections \ref{subsec:Our-contribution} and \ref{subsec:Our-technique}.

\textbf{Beyond solving linear programs}. Solving linear programs is
closely related to private learning subspaces and halfspaces \cite{DBLP:conf/focs/BunNSV15}.
\cite{DBLP:conf/colt/BeimelMNS19} show a reduction from learning
halfspaces to the problem of finding a point in the convex hull, where
an efficient algorithm for the latter problem implies an upper bound
for the sample complexity of efficient algorithms for the former.
Subsequent works by \cite{kaplan2020private,DBLP:conf/citc/GaoS21,DBLP:conf/nips/Ben-EliezerMZ22}
have targeted this question. \cite{kaplan2024differentially} build
techniques for solving this problem via privately solving LPs and
achieve the first algorithm that has polynomial dependence on the
dimension and polylogarithmic dependence on the domain size. Our algorithm
can be used as a subroutine to improve the runtime of the algorithm
by \cite{kaplan2024differentially}.

\section{Preliminaries\protect\label{sec:Preliminary}}

\textbf{Notation}. We let $\left\Vert \cdot\right\Vert _{1}$ be the
$\ell_{1}$-norm and $\left\Vert \cdot\right\Vert _{2}$ be the $\ell_{2}$-norm.
When it is clear from context, $\left\Vert \cdot\right\Vert $ also
denotes the $\ell_{2}$norm. For a matrix $A$, we denote by $a_{i}$
the $i$-th row of $A$. For a vector $a$, we let $\overline{a}$
be the normalized vector $\overline{a}=\frac{a}{\left\Vert a\right\Vert _{2}}$.
We also use $\overline{A}$ to denote the matrix $A$ with normalized
rows. We denote by $\Lap(b)$ the Laplace distribution with density
$f(x)=\frac{1}{2b}\exp(-\frac{|x|}{b})$; $N(0,\sigma^{2})$ the Gaussian
distribution with mean zero and variance $\sigma^{2}$ and $N(0,\sigma^{2}I)$
the multivariate Gaussian distribution with mean zero and covariance
$\sigma^{2}I$. The dimension will be clear from context.

\textbf{Linear programs}. We consider the problem of finding a feasible
solution to a linear program in general standard form $Ax\le b,x\geq0$,
where $A$ has dimension $n\times d$, $b$ is a vector of dimension
$n$ and the entries of $A$ and $b$ are integers. Following \cite{dunagan2004polynomial},
we refer to the problem of finding a solution satisfying $Ax\ge0$,
$x\neq0$ as a homogeneous LP. A homogeneous LP $Ax\ge0,x\neq0$ is
characterized by a quantity $\rho(A)$, namely, the margin (or roundness)
parameter, given as 
\begin{align*}
\rho(A) & =\max_{\left\Vert x\right\Vert _{2}\le1}\min_{i}\left\langle \overline{a}_{i},x\right\rangle .
\end{align*}
Geometrically, $\rho(A)$ is the radius of a ball that fits into the
intersection between the feasible region and the unit ball. The classic
perceptron algorithm for LPs with $\rho(A)>0$ converges with $1/\rho(A)^{2}$
iterations. Rescaling algorithms such as \cite{dunagan2004polynomial,soheili2013deterministic,hoberg2017improved}
have total runtime $\mathrm{poly}(n,d)\log1/\rho(A)$. 

\textbf{Homogenization}. \cite{dunagan2004polynomial} give a simple
reduction (called homogenization) from a general LP $Ax\le b,x\geq0$
to a homogeneous LP $A'x\ge0$, $x\neq0$ by setting $A'=\left[\begin{array}{c}
-A\mid b\\
I
\end{array}\right]$ and $x=(x\mid x_{0})^{\top}$. We refer to the homogeneous LP constructed
via this reduction as the homogenized LP. 

\textbf{Differential privacy}. We use the notation $(A,b)$ as shorthand
for the LP $Ax\leq b,x\geq0$. We say that two LPs $(A,b)$ and $(A',b')$
are neighbors is they differ by only one constraint (one LP has an
extra constraint). A randomized algorithm ${\cal A}$ is said to be
$(\epsilon,\delta)$-differentially private (DP) if for all neighboring
LPs $(A,b)$ and $(A',b')$ and every subset of possible outcomes
${\cal O}$,
\begin{align*}
\Pr\left[{\cal A}(A,b)\in{\cal O}\right] & \le e^{\epsilon}\Pr\left[{\cal A}(A',b')\in{\cal O}\right]+\delta.
\end{align*}
In the case $\delta=0$, we say the algorithm is $\epsilon$-DP. Two
commonly used mechanisms for achieving differential privacy are the
Laplace mechanism and the Gaussian mechanism. Let $f$ be a function
whose output has dimension $d$. We say $f$ has $\ell_{1}$ sensitivity
$k$ if on any two neighboring inputs $x$ and $x'$, $\left\Vert f(x)-f(x')\right\Vert _{1}\le k$,
and $f$ has $\ell_{2}$ sensitivity $k$ if on any two neighboring
inputs $x$ and $x'$, $\left\Vert f(x)-f(x')\right\Vert _{2}\le k$.
\begin{thm}[Laplace mechanism \cite{dwork2006calibrating}]
Let $f$ be a function of $\ell_{1}$ sensitivity $k$. The mechanism
${\cal A}$ that on an input $x$ adds independently generated noise
with the Laplace distribution $\Lap(\frac{k}{\epsilon})$ to each
of the $d$ coordinates of $f(x)$ is $\epsilon$-DP.
\end{thm}
\begin{thm}[Gaussian mechanism \cite{dwork2006our}]
Let $f$ be a function of $\ell_{2}$ sensitivity $k$. The mechanism
${\cal A}$ that on an input $x$ adds noise generated with the Gaussian
distribution $N(0,\sigma^{2})$ where $\sigma\ge\frac{k}{\epsilon}\sqrt{2\ln\frac{2}{\delta}}$
to each of the $d$ coordinates of $f(x)$ is $(\epsilon,\delta)$-DP.
\end{thm}

\section{Private Perceptron Algorithm for Positive Margin LPs\protect\label{sec:Private-Perceptron-Algorithm}}

\subsection{Algorithm}

\begin{algorithm}
\caption{Private Perceptron}

\label{alg:private-perceptron}

\begin{algorithmic}[1]

\STATE Input: Matrix $A\in\R^{n\times d}$, parameters $\rho_{0},\epsilon,\delta,\beta$

\STATE Let $B=I$

\STATE for $t=1\dots\tau=O\left(d\log\frac{1}{\rho_{0}}\right)$:

\STATE $\qquad$Run $c\gets\PPE(A,\epsilon,\delta,\beta)$

\STATE $\qquad$if $c$ is a solution:

\STATE $\qquad\qquad$return $Bc$

\STATE $\qquad$else if $c$ is a rescaling direction:

\STATE $\qquad\qquad$$A\gets A\left(I-\frac{1}{2}\overline{c}\cdot\overline{c}^{\top}\right)$;
$B\gets\left(I-\frac{1}{2}\overline{c}\cdot\overline{c}^{\top}\right)B$

\STATE $\qquad$else:

\STATE $\qquad\qquad$abort; return $\bot$

\STATE return $\bot$

\end{algorithmic}
\end{algorithm}

\begin{algorithm}
\caption{$\protect\PPE(A,\epsilon,\delta,\beta)$}

\label{alg:private-perceptron-epoch}

\begin{algorithmic}[1]

\STATE $x^{(1)}=\left(0,\dots,0\right)\in\R^{d}$

\STATE $\lambda^{(1)}=\left(0,\dots,0\right)\in\R^{n}$

\STATE $\nu=\Theta\left(\frac{\sqrt{d}}{\epsilon}\log^{1.5}\frac{T\tau}{\beta\delta}\right)$,
$\sigma^{2}=\frac{8\log\frac{4T}{\delta}}{\nu^{2}\epsilon^{2}}$,
$\theta^{2}=\frac{8}{d^{2}\epsilon^{2}\nu^{2}}\log\frac{4000\log\frac{\tau}{\beta}}{\delta}$

\STATE for $t=1\dots T=\Theta(d^{2})$: \hfill{}Perceptron Phase

\STATE $\qquad$$S^{(t)}=\left\{ i\in[n],\left\langle \overline{a}_{i},\overline{x}^{(t)}\right\rangle \le0\right\} $,
$m^{(t)}=\left|S^{(t)}\right|$

\STATE $\qquad$Let $\widehat{m}^{(t)}=m^{(t)}+\Lap\left(\frac{1}{\epsilon}\right)-\frac{1}{\epsilon}\log\frac{2000T\log\frac{1}{\beta}}{\delta}$.
If $\widehat{m}^{(t)}\le\nu$ then return solution $x^{(t)}$.

\STATE $\qquad$$u^{(t)}=\frac{1}{m^{(t)}}\sum_{i\in S^{(t)}}\overline{a}_{i}$,
$\widehat{u}^{(t)}=u^{(t)}+\eta^{(t)}$ where $\eta^{(t)}\sim N(0,\sigma^{2}I)$

\STATE $\qquad$$x^{(t+1)}=x^{(t)}+\widehat{u}^{(t)}$

\STATE $\qquad$$\lambda_{i}^{(t+1)}=\lambda_{i}^{(t)}+\frac{1}{m^{(t)}}$
for all $i\in S^{(t)}$ \hfill{}$\lambda^{(t)}$ are kept private

\STATE Let $\overline{\lambda}=\frac{\lambda^{(T)}}{T}$

\STATE for $s=1,\dots,1000\log\frac{\tau}{\beta}$:\label{line:applification}
\hfill{} Rescaling Phase

\STATE $\qquad$Take a gaussian vector $g^{(s)}\sim N(0,I)$ and
compute $P^{(s)}=\left\{ i:\left\langle \overline{a}_{i},g^{(s)}\right\rangle \ge0\right\} $

\STATE $\qquad$ Let $c^{(s)}=\sum_{i\in P^{(s)}}\overline{\lambda}_{i}\overline{a}_{i}$
and $\widehat{c}^{(s)}=\sum_{i\in P^{(s)}}\overline{\lambda}_{i}\overline{a}_{i}+\gamma^{(s)}$
where $\gamma^{(s)}\sim N(0,\theta^{2}I)$

\STATE $\qquad$if $\left\Vert \widehat{c}^{(s)}\right\Vert \ge\frac{3}{16\sqrt{\pi d}}$,
return $\widehat{c}^{(s)}$

\STATE Output $\bot$

\end{algorithmic}
\end{algorithm}

We describe our Private Perceptron algorithm in Algorithm \ref{alg:private-perceptron}.
Given a matrix $A\in\R^{n\times d}$, margin parameter $\rho_{0}$,
privacy parameters $\epsilon,\delta$, and failure probability $\beta$,
the algorithm runs at most $\tau=O(d\log\frac{1}{\rho_{0}})$ and
makes calls to $\PPE$ procedure given in Algorithm \ref{alg:private-perceptron-epoch}.
Algorithm \ref{alg:private-perceptron-epoch} has three possible outcomes.
If it outputs a solution, Algorithm \ref{alg:private-perceptron}
terminates and returns this solution with a suitable rescaling. If
it outputs a rescaling direction, Algorithm \ref{alg:private-perceptron}
rescales the input matrix $A$ and repeats. Otherwise, Algorithm \ref{alg:private-perceptron}
terminates and returns $\bot$.

Our main novel contribution lies in Algorithm \ref{alg:private-perceptron-epoch}.
This algorithm consists of two phases: the Perceptron Phase in which
the algorithm attempts to find a solution to the LP and the Rescaling
Phase in which the algorithm finds a good direction to rescale the
input if the solution from the Perceptron Phase is not satisfactory.
In each phase, we add maintain the privacy by adding appropriate noise.
During the Perceptron Phase, the algorithm maintains a solution and
updates it along the direction of a noisy average of the violated
constraints. The algorithm also maintains a weight vector $\lambda^{(t)}$
which picks up the constraints violated by the current solution. We
keep $\lambda^{(t)}$ private, and use the average value $\overline{\lambda}$
to determine the rescaling direction. To determine the rescaling direction,
during the Rescaling Phase, the algorithm takes a random gaussian
vector and computes a noisy sum of all rows satisfied by this vector
weighted by $\overline{\lambda}$. We will show that with a constant
probability, this noisy sum provides a good rescaling direction, and
the algorithm repeats the process a number of iterations to boost
the success probability.

In the next subsections, we show the privacy and utility guarantees
of our algorithm.

\subsection{Privacy analysis\protect\label{subsec:Privacy-analysis}}

Throughout, we let $(A,b)$ and $(A',b')$ be two neighboring LPs.
 The corresponding computed terms in Algorithm \ref{alg:private-perceptron-epoch}
for $(A',b')$ are denoted with the extra prime symbol (for example,
$S^{(t)}$ and $S^{(t)}{}'$).
\begin{prop}
\label{prop:main-privacy-lemma}Algorithm \ref{alg:private-perceptron}
is $(\epsilon',\delta')$-DP for $\epsilon'=2d^{3}\log\frac{1}{\rho_{0}}\epsilon^{2}+\sqrt{2d^{3}\epsilon^{2}\log\frac{1}{\rho_{0}}\log\frac{1}{\delta}}$
and $\delta'=(d+1)\delta$.
\end{prop}
To show this Proposition, we show the following lemmas. 
\begin{lem}
\label{lem:privacy-perceptron}Each iteration of the Perceptron Phase
is $(\epsilon,\frac{\delta}{2T})$-DP.
\end{lem}
The following claim follows similarly to the proof of the $\mathsf{NoisyAvg}$
algorithm by \cite{nissim2016locating}. We include the proof in the
appendix.
\begin{claim}
\label{claim:perceptron-sensitivity}For all $t\in[T]$, $m^{(t)}$
has $\ell_{1}$ sensitivity $1$, and $u^{(t)}$ has $\ell_{2}$ sensitivity
$\frac{2}{m^{(t)}}$.
\end{claim}
\begin{proof}[Proof of Lemma \ref{lem:privacy-perceptron}]
For the purpose of analysis, in each iteration of the Perceptron
Phase, we define the output of the algorithm to be either the solution
$x^{(t)}$ or the new update vector $\widehat{u}^{(t)}$. Let $O,O'$
be the outputs of the iteration on $(A,b)$ and $(A',b')$ respectively,
and let $F$ be an arbitrary subset of $\R^{d}$. First, due the the
privacy of the Laplace mechanism,
\begin{align*}
\Pr\left[O=x^{(t)}\right] & =\Pr\left[\widehat{m}^{(t)}\le\nu\right]\le e^{\epsilon}\Pr\left[\widehat{m}^{(t)}{}'\le\nu\right]=e^{\epsilon}\Pr\left[O'=x^{(t)}\right]
\end{align*}
Next, consider the case where the output is an update vector $\widehat{u}^{(t)}$.
If $m^{(t)}<\nu$:
\begin{align*}
\Pr\left[O=\widehat{u}^{(t)}\right] & \le\Pr\left[\widehat{m}^{(t)}>m^{(t)}\right]\le\Pr\left[\Lap\left(\frac{1}{\epsilon}\right)>\frac{1}{\epsilon}\log\frac{2T}{\delta}\right]\le\frac{\delta}{2T}.
\end{align*}
If $m^{(t)}\ge\nu$, then $\sigma^{2}\ge\frac{8\log\frac{4T}{\delta}}{\left(m^{(t)}\right)^{2}\epsilon^{2}}$,
and due to the privacy of the Gaussian mechanism,
\begin{align*}
\Pr\left[O=\widehat{u}^{(t)}\wedge\widehat{u}^{(t)}\in F\right] & \le\frac{\delta}{2T}+e^{\epsilon}\Pr\left[O'=\widehat{u}^{(t)}\wedge\widehat{u}^{(t)}\in F\right].
\end{align*}
\end{proof}
\begin{lem}
\label{lem:privacy-rescaling}Each iteration of the Rescaling Phase
is $(d\epsilon,\frac{\delta}{2000\log\frac{\tau}{\beta}})$-DP.
\end{lem}
To show this lemma, first, we show the sensitivity of $c^{(s)}$ in
the following claim, whose proof is deferred to the appendix.
\begin{claim}
\label{claim:rescaling-sensitivity}Let $m=\min_{t\in[T]}m^{(t)}$.
The $\ell_{2}$ sensitivity of $c^{(s)}$ is $\frac{2}{m}$.
\end{claim}
\begin{proof}[Proof of Lemma \ref{lem:privacy-rescaling}]
If $m<\nu$, the Rescaling Phase only happens when for all $t\in[T]$,
$\widehat{m}^{(t)}>\nu$. Hence, for any set of outcomes $F\subseteq\R^{d}$,
by union bound
\begin{align*}
 & \Pr\left[\text{Rescaling happens \ensuremath{\wedge}}\widehat{c}\in F\right]\le\Pr\left[\widehat{m}^{(t)}>\nu,\forall t\in T\right]\le\Pr\left[\widehat{m}^{(t)}>m,\forall t\in T\right]\\
 & \le\sum_{t\in[T]}\Pr\left[\widehat{m}^{(t)}>m^{(t)}\right]\le\sum_{t\in[T]}\Pr\left[\Lap\left(\frac{1}{\epsilon}\right)>\frac{1}{\epsilon}\log\frac{2000T\log\frac{\tau}{\beta}}{\delta}\right]\le\frac{\delta}{2000\log\frac{\tau}{\beta}}.
\end{align*}
Next, consider the case $m\ge\nu$. Since $\theta^{2}\ge\frac{8}{d^{2}\epsilon^{2}\nu^{2}}\log\frac{4000\log\frac{\tau}{\beta}}{\delta}$
by the privacy guarantee of the Gaussian mechanism
\begin{align*}
\Pr\left[\widehat{c}\in F\right] & \le\frac{\delta}{2000\log\frac{\tau}{\beta}}+e^{d\epsilon}\Pr\left[\widehat{c}'\in F\right].
\end{align*}
\end{proof}

\begin{proof}[Proof of Proposition \ref{prop:main-privacy-lemma}]
The algorithm is $(\epsilon',\delta')$-DP, following directly from
advanced composition.
\end{proof}

\subsection{Utility analysis\protect\label{subsec:Utility-analysis}}

To analyze the runtime and utility of Algorithm \ref{alg:private-perceptron}
we let $\B$ be the unit ball and $\P$ be the feasible region defined
by $Ax\ge0$. We will show the following proposition about the guarantee
on the output of Algorithm \ref{alg:private-perceptron}.
\begin{prop}
\label{prop:utility-guarantee}With probability at least $1-5\beta$,
Algorithm \ref{alg:private-perceptron} outputs a solution $x$ that
satisfies all but $O\left(\frac{\sqrt{d}}{\epsilon}\log^{1.5}\frac{d}{\beta\delta}\right)$
constraints.
\end{prop}
We outline the proof of Proposition \ref{prop:utility-guarantee}.
First, if in an iteration, Algorithm \ref{alg:private-perceptron}
finds a solution outputted by Algorithm \ref{alg:private-perceptron-epoch},
we show that this solution must be correct with probability $\ge1-\beta$.
\begin{lem}
\label{lem:solution-guarantee}If Algorithm \ref{alg:private-perceptron-epoch}
terminates in the Perceptron Phase and outputs a solution $x$, then
with probability at least $1-\beta$, $x$ satisfies all but $O\left(\frac{\sqrt{d}}{\epsilon}\log^{1.5}\frac{d}{\beta\delta}\right)$
constraints. 
\end{lem}
If Algorithm \ref{alg:private-perceptron} finds a rescaling vector
$c$ by the output of Algorithm \ref{alg:private-perceptron-epoch},
we show that the volume of the feasible region inside the unit ball
is increased by a constant factor with high probability. To start
with, assuming the initial margin parameter is at least $\rho_{0}$,
\cite{soheili2013deterministic} give a lower bound on the volume
of the initial feasible region inside the unit ball:
\begin{lem}[Lemma 3 \cite{soheili2013deterministic}]
 \label{lem:initital-volume}Suppose $\max_{x\in\P\cap\B}\min_{i}\left\langle \overline{a_{i}},x\right\rangle \ge\rho_{0}$
then $\vol\left(\mathcal{P}\cap\mathcal{B}\right)=\Omega\left(\rho_{0}^{d}\right)\vol(\B)$.
\end{lem}
Note that the rescaling operation $A\gets A\left(I-\frac{1}{2}\overline{c}\cdot\overline{c}^{\top}\right)$
is equivalent to a linear map $F:\R^{d}\to\R^{d}$ such that, $F_{c}(c)=2c$
and $F_{c}\left(x\right)=x$ for all $x\bot c$. \cite{hoberg2017improved}
show the following lemma.
\begin{lem}[Lemma 4 \cite{soheili2013deterministic}]
 \label{lem:increase-volume}Suppose that $c$ satisfies $\frac{1}{\left\Vert c\right\Vert }\max_{x\in{\cal {\cal P}}\cap\B}\left\langle c,x\right\rangle \le\frac{2}{3\sqrt{d}}$,
then $\vol\left(F(\P)\cap\B\right)\ge1.02\cdot\vol\left(\P\cap\B\right)$. 
\end{lem}
Next, we show that the rescaling vector $c$ outputted by Algorithm
\ref{alg:private-perceptron-epoch} satisfies the condition of Lemma
\ref{lem:increase-volume} with high probability.
\begin{lem}
\label{lem:rescaling-gaurantee}If Algorithm \ref{alg:private-perceptron-epoch}
does not return a solution, then with probability at least $1-\frac{4\beta}{\tau}$,
it outputs a rescaling vector $c$ that satisfies $\frac{1}{\left\Vert c\right\Vert }\max_{x\in{\cal {\cal P}}\cap\B}\left\langle c,x\right\rangle \le\frac{2}{3\sqrt{d}}$.
\end{lem}
Equipped with Lemma \ref{lem:solution-guarantee} - \ref{lem:rescaling-gaurantee},
we are now ready prove Proposition \ref{prop:utility-guarantee}.
\begin{proof}[Proof of Proposition \ref{prop:utility-guarantee}]
. The algorithm fails if either it outputs a solution that does not
satisfy more than $\Omega\left(\frac{\sqrt{d}}{\epsilon}\log^{1.5}\frac{d}{\beta\delta}\right)$
constraints, or it fails to output a rescaling vector that satisfies
the condition of Lemma \ref{lem:increase-volume}. This happens with
probability at most $\beta+\tau\cdot\frac{4\beta}{\tau}=5\beta$.

Otherwise, in each iteration, either the algorithm outputs a satisfactory
solution, or the volume of the feasible region inside the unit ball
is increased by a factor at least $1.02$, by Lemma \ref{lem:increase-volume}.
Therefore, by Lemma \ref{lem:initital-volume}, the algorithm stops
after $O\left(\log\frac{1}{\rho_{0}^{d}}\right)=O\left(d\log\frac{1}{\rho_{0}}\right)$.
\end{proof}
The remaining work is to prove Lemmas \ref{lem:solution-guarantee}
and \ref{lem:rescaling-gaurantee}.

\begin{proof}[Proof of Lemma \ref{lem:solution-guarantee}]
If Algorithm \ref{alg:private-perceptron-epoch} terminates in iteration
$t$ of the Perceptron Phase, we have $\widehat{m}^{(t)}\le\nu$ where
$\nu=O\left(\frac{\sqrt{d}}{\epsilon}\log^{1.5}\frac{d}{\beta\delta}\right)$.
That is
\begin{align*}
m^{(t)} & \le\nu-\Lap\left(\frac{1}{\epsilon}\right)+\frac{1}{\epsilon}\log\frac{2000T\log\frac{1}{\beta}}{\delta}
\end{align*}
Here, $m^{(t)}$ is the number of constraints not satisfied by the
solution, and we have
\begin{align*}
\Pr\left[m^{(t)}\ge\nu+\frac{\log\frac{1}{\beta}}{\epsilon}+\frac{1}{\epsilon}\log\frac{2000T\log\frac{1}{\beta}}{\delta}\right] & \le\Pr\left[\Lap\left(\frac{1}{\epsilon}\right)\le-\frac{\log\frac{1}{\beta}}{\epsilon}\right]\le\beta.
\end{align*}
\end{proof}

To show Lemma \ref{lem:rescaling-gaurantee} we start with the following
claim:
\begin{claim}
\label{claim:solution-length-bound}$\Pr\left[\left\Vert x^{(t)}\right\Vert \le10\sqrt{t},\ \forall t\le T\right]\ge1-\frac{2\beta}{\tau}$.
\end{claim}
To show Claim \ref{claim:solution-length-bound}, we use the following
facts about Gaussian random variables.
\begin{fact}
\label{fact}If $X\sim N(0,\sigma^{2})$ then for $a\ge0$, $\Pr\left[X\ge\sigma a\right]\le\exp\left(-a^{2}/2\right)$.
If $X\sim N(0,\sigma^{2}I)\in\R^{d}$ then for $u\in\R^{d}$, $\left\langle u,X\right\rangle $
follows $N(0,\sigma^{2}\left\Vert u\right\Vert ^{2})$ and $\frac{1}{\sigma^{2}}\left\Vert X\right\Vert ^{2}$
follows $\chi$-squared distribution with $d$ degrees of freedom.
Furthermore, for $a\ge0$, $\Pr\left[\frac{1}{\sigma^{2}}\left\Vert X\right\Vert ^{2}\ge d+2\sqrt{da}+2a\right]\le\exp\left(-a\right)$.
\end{fact}
\begin{proof}[Proof of Claim \ref{claim:solution-length-bound}]
 We give an outline of this proof. First, notice that, by the update
of $x^{(t)}$, we have with high probability (proof will follow):
\begin{align*}
\text{\ensuremath{\left\Vert x^{(t+1)}\right\Vert ^{2}}} & \le\left\Vert x^{(t)}\right\Vert ^{2}+2\left\langle x^{(t)},\eta^{(t)}\right\rangle +3.
\end{align*}
Then, the main task is to bound $\sum_{s}\left\langle x^{(s)},\eta^{(s)}\right\rangle $.
However, $x^{(s)}$ is not a bounded variable, so we cannot directly
apply concentration inequalities. Instead, we will define a bounded
sequence $(Y^{(s)})$ that behaves like $(\left\langle x^{(s)},\eta^{(s)}\right\rangle )$
and proceed to bound $\sum_{s=1}^{t}Y^{(s)}$. Finally, we will show
that with high probability $(Y^{(s)})$ behaves like $(\left\langle x^{(s)},\eta^{(s)}\right\rangle )$
and from there we can bound $\left\Vert x^{(t)}\right\Vert $.

For each iteration $t$ of Algorithm \ref{alg:private-perceptron-epoch},
consider the following event, called $M_{t}$ that both of the following
happen: 1)$\left\langle u^{(t)},\eta^{(t)}\right\rangle \le\sigma\sqrt{2\log\frac{2T\tau}{\beta}}\le\frac{1}{2}$,
and 2)$\left\Vert \eta^{(t)}\right\Vert _{2}^{2}\le5\sigma^{2}d\log\frac{2T\tau}{\beta}\le\frac{1}{\log\frac{2T}{\beta}}\le1$.
Since $u^{(t)}$ is the average of unit length vectors, $\left\Vert u^{(t)}\right\Vert \le1$.
Then because $\eta^{(t)}\sim N(0,\sigma^{2}I)$, $\left\langle u^{(t)},\eta^{(t)}\right\rangle $
follows a $N(0,\kappa^{2})$ with $\kappa^{2}\le\sigma^{2}$ (Fact
\ref{fact}). Additionally, $\frac{1}{\sigma^{2}}\left\Vert \eta^{(t)}\right\Vert ^{2}$
follows a $\chi$-squared distribution with $d$ degrees of freedom
(Fact \ref{fact}). Therefore, by Fact \ref{fact},
\begin{align*}
\Pr\left[M_{t}\right] & \ge1-\left(\Pr\left[\left\langle u^{(t)},\eta^{(t)}\right\rangle >\sigma\sqrt{2\log\frac{2T\tau}{\beta}}\right]+\Pr\left[\left\Vert \eta^{(t)}\right\Vert ^{2}>5\sigma^{2}d\log\frac{2T\tau}{\beta}\right]\right)\ge1-\frac{\beta}{T\tau}.
\end{align*}
We define the following random variable:
\begin{align*}
Y^{(t)} & =\begin{cases}
\left\langle x^{(t)},\eta^{(t)}\right\rangle  & \text{if }\cap_{s=0}^{t}M_{t}\text{ happens and }\left\Vert x^{(t)}\right\Vert \le10\sqrt{t}\\
0 & \text{otherwise}
\end{cases}
\end{align*}
Due to the symmetry of $\eta^{(t)}$, $\E\left[Y^{(t)}\mid\F_{t-1}\right]=0$.
If $\cap_{s=1}^{t}M_{t}$ happens and $\left\Vert x^{(t)}\right\Vert \le10\sqrt{t}$
we have 
\begin{align*}
\left|Y^{(t)}\right| & =\left|\left\langle x^{(t)},\eta^{(t)}\right\rangle \right|\le\left\Vert x^{(t)}\right\Vert \left\Vert \eta^{(t)}\right\Vert \le10\sqrt{\frac{t}{\log\frac{2T\tau}{\beta}}}
\end{align*}
Then $(Y^{(t)})$ forms a bounded martingale difference sequence.
By Azuma's inequality,for all $t$
\begin{align*}
\Pr\left[\left|\sum_{s=1}^{t}Y^{(s)}\right|\ge40t\right] & \le2\exp\left(-\frac{(40t)^{2}}{2\sum_{s=1}^{t}\left(20\sqrt{\frac{s}{\log\frac{2T\tau}{\beta}}}\right)^{2}}\right)\le2\exp\left(-\frac{1600t^{2}\log\frac{2T\tau}{\beta}}{800t^{2}}\right)\le\frac{\beta}{T\tau}.
\end{align*}
Thus by union bound we have with probability $1-2\beta/\tau$, $M_{t}$
happens and $\left|\sum_{s=1}^{t}Y^{(s)}\right|\le40t$, for all $t$,
simultaneously. When this happens we show by induction that $\left\Vert x^{(t)}\right\Vert \le10\sqrt{t}$.
This is true for $t=1$. Suppose that we have $\left\Vert x^{(s)}\right\Vert \le10\sqrt{s}$
for all $s\le t$. We have
\begin{align*}
\text{\ensuremath{\left\Vert x^{(t+1)}\right\Vert ^{2}}} & =\left\Vert x^{(t)}+u^{(t)}+\eta^{(t)}\right\Vert ^{2}\\
 & \le\left\Vert x^{(t)}\right\Vert ^{2}+\left\Vert u^{(t)}\right\Vert ^{2}+\left\Vert \eta^{(t)}\right\Vert ^{2}+2\left\langle x^{(t)},u^{(t)}\right\rangle +2\left\langle u^{(t)},\eta^{(t)}\right\rangle +2\left\langle x^{(t)},\eta^{(t)}\right\rangle \\
 & \le\left\Vert x^{(t)}\right\Vert ^{2}+2\left\langle x^{(t)},\eta^{(t)}\right\rangle +3.
\end{align*}
where in the last inequality, conditioned on $M_{t}$, we have $\left\Vert \eta^{(t)}\right\Vert ^{2}\le1$
and $\left\langle u^{(t)},\eta^{(t)}\right\rangle \le\frac{1}{2}$.
Since $u^{(t)}=\frac{1}{m^{(t)}}\sum_{i\in S^{(t)}}\overline{a}_{i}$,
and $S^{(t)}=\left\{ i\in[n],\left\langle \overline{a}_{i},\overline{x}^{(t)}\right\rangle \le0\right\} $,
we have $\left\Vert u^{(t)}\right\Vert ^{2}\le1$ and $\left\langle x^{(t)},u^{(t)}\right\rangle \le0$.
Continuing expanding this recursion, we have 
\begin{align*}
\text{\ensuremath{\left\Vert x^{(t+1)}\right\Vert ^{2}}} & \le\underbrace{\left\Vert x^{(1)}\right\Vert ^{2}}_{=0}+2\sum_{s=1}^{t}\left\langle x^{(s)},\eta^{(s)}\right\rangle +3t.
\end{align*}
Due to the induction hypothesis $\left\Vert x^{(s)}\right\Vert \le10\sqrt{s}$
for all $s\le t$, and $M_{t}$ happens for all $t$ we have then
$Y^{(s)}=\left\langle x^{(s)},\eta^{(s)}\right\rangle $ for all $s\le t$.
It follows that $2\sum_{s=1}^{t}\left\langle x^{(s)},\eta^{(s)}\right\rangle \le2\sum_{s=1}^{t}Y^{(s)}\le40t$.
Therefore 
\begin{align*}
\text{\ensuremath{\left\Vert x^{(t+1)}\right\Vert ^{2}}} & \le2\cdot40t+3t\le100(t+1).
\end{align*}
\end{proof}

\begin{claim}
\label{claim:coefficient-length-bound}With probability at least $1-\frac{3\beta}{\tau}$,
$\left\Vert \overline{\lambda}\overline{A}\right\Vert \le\frac{11}{\sqrt{T}}$,
where for convenience we define $\overline{\lambda}\overline{A}=\diag(\overline{\lambda})\overline{A}$
and $\diag(\overline{\lambda})$ is the diagonal matrix obtained from
$\overline{\lambda}$.
\end{claim}
\begin{proof}
First, with probability at least $1-\frac{\beta}{\tau},$we have $\left\Vert \sum_{s=1}^{T}\eta^{(s)}\right\Vert _{2}^{2}\le5\sigma^{2}Td\log\frac{T\tau}{\beta}\le T$.
By Claim \ref{claim:solution-length-bound} and union bound, we have,
with probability at least $1-\frac{3\beta}{\tau}$, $\left\Vert \sum_{s=1}^{T}\eta^{(s)}\right\Vert \le\sqrt{T}$
and $\left\Vert x^{(T)}\right\Vert \le10\sqrt{T}$. We have that 
\begin{align*}
\left(\lambda^{(t+1)}-\lambda^{(t)}\right)\overline{A} & =u^{(t)}=\widehat{u}^{(t)}-\eta^{(t)}=x^{(t+1)}-x^{(t)}-\eta^{(t)}
\end{align*}
Hence$\lambda^{(T)}\overline{A}=x^{(T)}-\sum_{s=1}^{T-1}\eta^{(s)}$.
Thus $\left\Vert \lambda^{(T)}\overline{A}\right\Vert \le\left\Vert x^{(T)}\right\Vert +\left\Vert \sum_{s=1}^{T-1}\eta^{(T)}\right\Vert \le11\sqrt{T}.$
The claim follows due to $\overline{\lambda}=\frac{\lambda^{(T)}}{T}$.
\end{proof}

\begin{claim}
\label{claim:boosting}Conditioned on $\left\Vert \overline{\lambda}\overline{A}\right\Vert \le\frac{11}{\sqrt{T}}$,
with probability at least $10^{-3}$, we have $\left\Vert \widehat{c}^{(s)}\right\Vert \ge\frac{3}{16\sqrt{\pi d}}$,
in which case $\frac{1}{\left\Vert \widehat{c}^{(s)}\right\Vert }\max_{x\in{\cal {\cal P}}\cap\B}\left\langle \widehat{c}^{(s)},x\right\rangle \le\frac{2}{3\sqrt{d}}$.
\end{claim}
In order to prove Claim \ref{claim:boosting}, we need Lemma 6 from
\cite{hoberg2017improved}.
\begin{lem}[Lemma 6 \cite{hoberg2017improved}]
\label{lem:rothvoss-6}Let $\lambda\in\R^{n}\ge0$ and $A\in\R^{n\times d}$
be such that $\sum_{i}\lambda_{i}=1$ and $\left\Vert A_{i}\right\Vert =1$.
Take a random gaussian vector $g\sim N(0,I)\in\R^{d}$ and let $J=\{i:\left\langle A_{i},g\right\rangle \ge0\}$.
With probability at least $5\cdot10^{-3}$, $\left\Vert \sum_{i\in J}\lambda_{i}A_{i}\right\Vert \ge\frac{1}{4\sqrt{\pi}d}$.
\end{lem}
\begin{proof}
For simplicity, we drop the superscript $s$. The proof follows similarly
to \cite{hoberg2017improved}; however, we need to take into account
the noise component $\gamma$. Since $\gamma\sim N\left(0,\theta^{2}I\right)$
with $\theta^{2}=\frac{8}{d^{2}\epsilon^{2}\nu^{2}}\log\frac{4000\log\frac{\tau}{\beta}}{\delta}\le O\left(\frac{1}{d^{3}\log\frac{Td}{\beta\delta}}\right)$,
$\frac{1}{\theta^{2}}\left\Vert \gamma\right\Vert ^{2}$ follows a
$\chi$-squared distribution of $d$ degrees of freedom (Fact \ref{fact}).
By Fact \ref{fact}, for sufficiently large constant in $\nu$, $\Pr\left[\left\Vert \gamma\right\Vert \le\frac{1}{16\sqrt{\pi}d}\right]\ge1-\frac{\beta}{\tau}$.
Further, by Lemma \ref{lem:rothvoss-6}, with probability $\ge5\cdot10^{-3}$
we have $\left\Vert \sum_{i\in P}\overline{\lambda}_{i}\overline{a}_{i}\right\Vert \ge\frac{1}{4\sqrt{\pi d}}$.

The two events: $\left\Vert \sum_{i\in P}\overline{\lambda}_{i}\overline{a}_{i}\right\Vert \ge\frac{1}{4\sqrt{\pi d}}$
and $\left\Vert \gamma\right\Vert \le\frac{1}{16\sqrt{\pi}d}$ are
independent. With probability at least $(1-\frac{\beta}{\tau})\cdot5\cdot10^{-3}\ge10^{-3}$,
both events happen. Then,
\begin{align*}
\left\Vert \widehat{c}\right\Vert =\left\Vert \sum_{i\in P}\overline{\lambda}_{i}\overline{a}_{i}+\gamma\right\Vert  & \ge\left\Vert \sum_{i\in P}\overline{\lambda}_{i}\overline{a}_{i}\right\Vert -\left\Vert \gamma\right\Vert \ge\frac{3}{4}\left\Vert \sum_{i\in P}\overline{\lambda}_{i}\overline{a}_{i}\right\Vert \ge\frac{3}{16\sqrt{\pi d}}.
\end{align*}
To complete the proof, we now show that, conditioned on the two events
$\left\Vert \sum_{i\in P}\overline{\lambda}_{i}\overline{a}_{i}\right\Vert \ge\frac{1}{4\sqrt{\pi d}}$
and $\left\Vert \gamma\right\Vert \le\frac{1}{16\sqrt{\pi}d}$ happening,
we have $\frac{1}{\left\Vert \widehat{c}^{(s)}\right\Vert _{2}}\max_{x\in{\cal {\cal P}}\cap\B}\left\langle \widehat{c}^{(s)},x\right\rangle \le\frac{2}{3\sqrt{d}}$.
Since $\left\Vert \widehat{c}\right\Vert \ge\frac{3}{16\sqrt{\pi d}}$
and $\left\Vert \gamma\right\Vert \leq\frac{1}{16\sqrt{\pi}d}\leq\frac{1}{3}\left\Vert \widehat{c}\right\Vert $,
we have $\left\Vert \sum_{i\in P}\overline{\lambda}_{i}\overline{a}_{i}\right\Vert \ge\left\Vert \widehat{c}\right\Vert -\left\Vert \gamma\right\Vert \ge\frac{1}{8\sqrt{\pi d}}$
and $\left\Vert \sum_{i\in P}\overline{\lambda}_{i}\overline{a}_{i}\right\Vert \le\left\Vert \widehat{c}\right\Vert +\left\Vert \gamma\right\Vert \le\frac{4}{3}\left\Vert \widehat{c}\right\Vert $.
This gives us
\begin{align*}
\frac{1}{\left\Vert \widehat{c}\right\Vert }\max_{x\in{\cal {\cal P}}\cap\B}\left\langle \widehat{c},x\right\rangle  & =\frac{1}{\left\Vert \widehat{c}\right\Vert }\max_{x\in{\cal {\cal P}}\cap\B}\left\langle \sum_{i\in P}\overline{\lambda}_{i}\overline{a}_{i}+\gamma,x\right\rangle \\
 & \le\frac{4}{3\left\Vert \sum_{i\in P}\overline{\lambda}_{i}\overline{a}_{i}\right\Vert }\max_{x\in\P\cap\B}\left\langle \sum_{i\in P}\overline{\lambda}_{i}\overline{a}_{i},x\right\rangle +\frac{16\sqrt{\pi d}}{3}\left\Vert \gamma\right\Vert \left\Vert x\right\Vert \\
 & \le\frac{4}{3}\frac{\left\Vert \overline{\lambda}\overline{A}\right\Vert }{\left\Vert \sum_{i\in P}\overline{\lambda}_{i}\overline{a}_{i}\right\Vert }+\frac{16\sqrt{\pi d}}{3}\frac{1}{16\sqrt{\pi}d}\\
 & \le\frac{4\cdot11\cdot8\sqrt{\pi d}}{3\sqrt{T}}+\frac{1}{3\sqrt{d}}\le\frac{2}{3\sqrt{d}}.
\end{align*}
The third inequality is due to $\left\Vert \gamma\right\Vert \leq\frac{1}{16\sqrt{\pi}d}$
and $\left\Vert x\right\Vert \leq1$ since $x\in\mathcal{B}$. Besides,
since $Ax\ge0$ for $x\in\P$, we also have
\begin{align*}
\max_{x\in\P\cap\B}\left\langle \sum_{i\in P}\overline{\lambda}_{i}\overline{a}_{i},x\right\rangle  & \le\max_{x\in\P\cap\B}\left\langle \sum_{i\in[n]}\overline{\lambda}_{i}\overline{a}_{i},x\right\rangle \le\max_{\left\Vert x\right\Vert \le1}\left\langle \sum_{i\in[n]}\overline{\lambda}_{i}\overline{a}_{i},x\right\rangle =\left\Vert \overline{\lambda}\overline{A}\right\Vert ,
\end{align*}
The last inequality holds for large enough $T=\Theta(d^{2})$.
\end{proof}

\begin{proof}[Proof of Lemma \ref{lem:rescaling-gaurantee}]
The proof of Lemma \ref{lem:rescaling-gaurantee} follows immediately
from Claims \ref{claim:coefficient-length-bound} and \ref{claim:boosting}
when we repeat $O(\log\frac{\tau}{\beta})$ times in the for-loop
(Line \ref{line:applification} of Algorithm \ref{alg:private-perceptron-epoch}).
\end{proof}

\subsection{Putting it all together}
\begin{proof}[Proof of Theorem \ref{thm:perceptron-guarantee}]
 Theorem \ref{thm:perceptron-guarantee} follows from Propositions
\ref{prop:main-privacy-lemma} and \ref{prop:utility-guarantee} where
we set $\delta\gets\frac{\delta}{d}$ and $\epsilon=\Theta\left(\sqrt{\frac{1}{d^{3}\log\frac{1}{\rho_{0}}\log\frac{d}{\delta}}}\right)$.
\end{proof}

\section{When the LP is Not Fully Dimensional\protect\label{sec:Zero-Margin}}

Algorithm \ref{alg:private-perceptron} can only guarantee the number
of dropped constraints is $\tilde{O}(d^{2})$ when the LP has a strictly
positive margin. In case the LP has zero margin, i.e, tight (equality)
constraints, we can use a standard perturbation technique to create
a margin without changing the feasibility of the problem. However,
the challenge is to output a solution that satisfies the original
constraints. The main idea to handle this case, due to \cite{kaplan2024differentially},
is to iteratively identify tight constraints based on the following
observation. An LP with integer entries bounded by $U$ in the absolute
value has the same feasibility as that of a relaxed LP with $\eta=\frac{1}{2(d+1)((d+1)U)^{d+1}}$
slackness added to each constraint and, if a solution to the relaxed
LP violates a constraint in the initial LP, that constraint must be
tight. This suggests that we can solve the relaxed LP which has a
positive margin, then identify tight constraints and recurse at most
$d$ times (we can stop after identifying $d$ linearly independent
tight constraints).

Our improvement over the algorithm of \cite{kaplan2024differentially}
is two-fold. First, we use Algorithm \ref{alg:private-perceptron}
as the solver in each iteration, which improves the number of dropped
constraints from $\tilde{O}(\frac{d^{5}}{\epsilon}\log^{1.5}\frac{1}{\rho_{0}})$
to $O\left(\frac{d^{2}}{\epsilon}\sqrt{\log\frac{1}{\rho_{0}}}\log^{2}\frac{d}{\beta\delta}\right)$.
Second, we avoid the fast decrease of the margin during the course
of the algorithm, which saves another factor $d$.

\subsection{Synthetic systems of linear equations}

During its course, the algorithm needs to identify the set of tight
constraints and privatize them for the subsequent use. To this end,
we will use the algorithm from \cite{kaplan2024differentially} for
privately generating (or \textit{sanitizing}) synthetic systems of
linear equations, stated in the following result.
\begin{thm}[Theorem C \cite{kaplan2024differentially}]
\label{thm:sanitize}There exists an $(\epsilon,\delta)$-DP algorithm
such that for any system of linear equations ${\cal I}$, with probability
at least $1-\delta$, the algorithm outputs a system of linear equations
${\cal O}$ which satisfies

1. Any solution to ${\cal I}$ is a solution to ${\cal O}$, and

2. Each solution to ${\cal O}$ satisfies all but $O\left(\frac{d^{2}}{\epsilon}\log\frac{d}{\delta}\right)$
equations in ${\cal I}.$
\end{thm}
Note that for the analysis we also use the fact that this algorithm
first privately selects a set of tight constraints in the original
LP, then outputs a canonical basis for the subspace defined by these
tight constraints.

\subsection{Algorithm}

The algorithm is as follows. The algorithm runs at most $d$ iterations.
In each iteration, the algorithm solves the LP $Ax\le b+\eta\one$
with a set $E$ of sanitized equality constraints that were identified
in previous iterations. First, the algorithm selects $k=\left|E\right|$
independent columns in $E$ to eliminate $k$ variables then solves
the LP without equality constraints using the Algorithm \ref{alg:private-perceptron}
(after homogenizing the LP via the reduction from Section \ref{sec:Preliminary}).
The algorithm terminates if it has found $d$ tight constraints or
if the number of constraints violated by the solution is small. Otherwise,
the algorithm sanitizes the set of constraints violated by the solution
and adds them to $E$ then recurses. Our algorithm is described in
Algorithm \ref{alg:general-case}. 

\begin{algorithm}
\caption{Private solver for $Ax\le b$}

\label{alg:general-case}

\begin{algorithmic}[1]

\STATE Input $A,b$ whose entries are integers bounded by $U$ in
absolute value, $\epsilon,\delta,\beta$

\STATE Let $\eta=\frac{1}{2(d+1)((d+1)U)^{d+1}}$, $\rho=\eta^{3}$

\STATE Let $E=\emptyset$ be the set of tight (equality) constraints

\STATE for $t=1,\dots,d$:

\STATE $\qquad$Use $k=\left|E\right|$ independent columns in $E$
to eliminate $k$ variables of $Ax\le b+\eta\one$ and use $\PP$
to solve with margin parameter $\rho$, obtain solution $x^{(t)}$\label{drop-perceptron}

\STATE $\qquad$Let $J_{1}$ be the subset of constraints $j$ such
that $\left\langle a_{j},x^{(t)}\right\rangle \le b_{j}$ and $J_{2}$
be the subset of constraints $j$ such that $b_{j}<\left\langle a_{j},x^{(t)}\right\rangle \le b_{j}+\eta$

\STATE $\qquad$if $\left|J_{2}\right|+\Lap\left(\frac{1}{\epsilon}\right)\le\frac{d^{2}}{\epsilon}+\log\frac{1}{\delta}$
return $x^{(t)}$\label{return-1}

\STATE $\qquad$else:

\STATE $\qquad\qquad$Let $Cx=g$ be the system obtained by sanitizing
$A_{J_{2}}x=b_{J_{2}}$ and let $s$ be its dimension

\STATE $\qquad\qquad$if $s=d$ or $\left|J_{1}\right|+\Lap\left(\frac{1}{\epsilon}\right)<\frac{d^{2}}{\epsilon}\log^{2}\frac{d}{\beta\delta}\sqrt{\log\frac{1}{\rho}}+\log\frac{1}{\delta}$,
return a solution of $Cx=g$\label{return-2}

\STATE $\qquad\qquad$else:

\STATE $\qquad\qquad\qquad$$E\gets E\cup\{Cx=g\}$; $A\gets A_{J_{1}}$;
$b\gets b_{J_{1}}$

\STATE return $\bot$

\end{algorithmic}
\end{algorithm}

\subsection{Analysis}

For the privacy guarantee, we have the following lemma.
\begin{prop}
\label{lem:general-privacy}The algorithm is $\left(\epsilon',\delta'\right)$-DP
for $\epsilon'=3d\epsilon^{2}+\sqrt{3d\epsilon^{2}\log\frac{1}{\delta}};$
$\delta'=(3d+1)\delta$.
\end{prop}
\begin{proof}
This comes directly from advanced composition over $d$ iterations,
each of which uses three $(\epsilon,\delta)$ private mechanisms. 
\end{proof}

Next, for the utility guarantee, we first show that during the iterations,
the entries of the LP (after variable elimination) does not blow up
by a factor more than $U^{d}$ and thereby, we can lower bound the
margin of the relaxed LP.

\begin{lem}
\label{lem:margin-bound}For each step $t$ of Algorithm \ref{alg:general-case},
consider the LP $Ax\le b$, $x\ge0$ with the entries of $A,b$ bounded
by $U$ in the absolute value and let $E$ be the set of $k$ equality
constraints $Cx=g$ where $C$ has rank $k$. We can use $k$ independent
columns in $C$ and Gaussian elimination to eliminate $k$ variables
to obtain an LP $\tilde{A}x\le\tilde{b}$, $x\ge0$ with integer entries
and $d-k$ variables such that the entries of $\tilde{A},\tilde{b}$
are bounded by $k^{2}(k-1)!U^{k+1}$. Furthermore, if the LP $\{Ax\le b;Cx=g,x\ge0\}$
is feasible then the LP resulting from the same elimination from $\{Ax\le b+\eta\one;Cx=g,x\ge0\}$
has margin parameter $\rho_{t}\ge\frac{\eta}{((d+1)U)^{2(d+1)}}$.
\end{lem}
\begin{proof}
See Appendix.
\end{proof}
Equipped with Lemma \ref{lem:margin-bound}, we can show the bound
for the number of dropped constraints.
\begin{prop}
\label{lem:general-utility}With probability at least $1-d(\beta+\delta)$,
Algorithm \ref{alg:general-case} outputs a solution that satisfies
all but $O(\frac{d^{3.5}}{\epsilon}\log^{2}\frac{d}{\beta\delta}\sqrt{\log(dU)})$
constraints.
\end{prop}
\begin{proof}
See Appendix.
\end{proof}

\subsection{Proof of Theorem \ref{thm:general-guarantee}}
\begin{proof}
The proof of Theorem \ref{thm:general-guarantee} follows from Propositions
\ref{lem:general-privacy} and \ref{lem:general-utility} where we
select $\delta\gets\frac{\delta}{d}$, $\epsilon\gets\frac{1}{\sqrt{d}\sqrt{\log\frac{1}{\delta}}}$
and $\beta\gets\frac{\beta}{d}$.
\end{proof}
\bibliographystyle{alpha}
\bibliography{ref}

\appendix

\section{Missing Proofs from Section \ref{sec:Private-Perceptron-Algorithm}}

\subsection{From Section \ref{subsec:Privacy-analysis}}

As a reminder, we let $(A,b)$ and $(A',b')$ be the neighboring LPs.
 The corresponding computed terms in Algorithm \ref{alg:private-perceptron-epoch}
for $V'$ are denoted with the extra prime symbol (for example, $S^{(t)}$
and $S^{(t)}{}'$).

\begin{proof}[Proof of Claim \ref{claim:perceptron-sensitivity}]
Fix two neighboring LPs $V$ and $V'$. Suppose $V'$ has one extra
constraint. $S^{(t)}{}'$ can have at most one more constraint so
it is immediate that $m^{(t)}$ has $\ell_{2}$ sensitivity $1$.
We consider the case $S^{(t)}{}'$ has one extra constraint $\overline{a}$,
\begin{align*}
u^{(t)}-u^{(t)}{}' & =\left(\frac{1}{m^{(t)}}\sum_{i\in S^{(t)}}\overline{a}_{i}\right)-\left(\frac{\overline{a}}{m^{(t)}+1}+\frac{1}{m^{(t)}+1}\sum_{i\in S^{(t)}}\overline{a}_{i}\right)=\frac{\sum_{i\in S^{(t)}}\overline{a}_{i}}{m^{(t)}(m^{(t)}+1)}-\frac{\overline{a}}{m^{(t)}+1}.
\end{align*}
Then by triangle inequality
\begin{align*}
\left\Vert u^{(t)}-u^{(t)}{}'\right\Vert  & \le\frac{\sum_{i\in S^{(t)}}\left\Vert \overline{a}_{i}\right\Vert }{m^{(t)}(m^{(t)}+1)}+\frac{\left\Vert \overline{a}\right\Vert }{m^{(t)}+1}=\frac{2}{m^{(t)}+1}\le\frac{2}{m^{(t)}}.
\end{align*}
When $V'$ has one less constraint, we proceed similarly.
\begin{align*}
u^{(t)}-u^{(t)}{}' & =\left(\frac{\overline{a}}{m^{(t)}}+\frac{1}{m^{(t)}}\sum_{i\in S^{(t)}{}'}\overline{a}_{i}\right)-\left(\frac{1}{m^{(t)}-1}\sum_{i\in S}\overline{a}_{i}\right)=\frac{\sum_{i\in S^{(t)}{}'}\overline{a}_{i}}{m^{(t)}(m^{(t)}-1)}+\frac{\overline{a}}{m^{(t)}}.
\end{align*}
Then by triangle inequality
\begin{align*}
\left\Vert u^{(t)}-u^{(t)}{}'\right\Vert  & \le\frac{\sum_{i\in S^{(t)}{}'}\left\Vert \overline{a}_{i}\right\Vert }{m^{(t)}(m^{(t)}-1)}+\frac{\left\Vert \overline{a}\right\Vert }{m^{(t)}}=\frac{2}{m^{(t)}}.
\end{align*}
\end{proof}

\begin{proof}[Claim \ref{claim:rescaling-sensitivity}]
For simplicity, we drop the superscript $s$ for iteration $s$.
Consider the case where $V'$ and $P'$ have one extra constraint
$a_{k}$. For all $t\in[T]$, for $i\neq k$ we have
\begin{align*}
0\le\lambda_{i}^{(t)}-\lambda_{i}'{}^{(t)} & \le\frac{1}{m^{(t)}}-\frac{1}{m^{(t)}+1}=\frac{1}{m^{(t)}\left(m^{(t)}+1\right)}\le\frac{\lambda_{i}'{}^{(t)}}{m^{(t)}}\le\frac{\lambda_{i}^{(t)}}{m}.
\end{align*}
Hence $\left|\overline{\lambda}_{i}-\overline{\lambda}_{i}'\right|\le\frac{\overline{\lambda}_{i}}{m}$.
Besides, we also have $\overline{\lambda}'_{k}\le\frac{1}{m}$. For
coordinate $j$,
\begin{align*}
c_{j}-c'_{j} & =\sum_{i\in P}\overline{\lambda}_{i}\overline{a}_{i,j}-\sum_{i\in P}\overline{\lambda}'_{i}\overline{a}_{i,j}-\overline{\lambda}'_{k}\overline{a}_{k,j}=\sum_{i\in P}\left(\overline{\lambda}_{i}-\overline{\lambda}'_{i}\right)\overline{a}_{i,j}-\overline{\lambda}'_{k}\overline{a}_{k,j}.
\end{align*}
Then, to compute the $\ell_{2}$ sensitivity of $c$, we take $\sum_{j}\left(c_{j}-c'_{j}\right)^{2}$.
\begin{align*}
\sum_{j}\left(c_{j}-c'_{j}\right)^{2} & =\sum_{j}\left(\sum_{i\in P}\underbrace{\left(\overline{\lambda}_{i}-\overline{\lambda}'_{i}\right)}_{\ge0}\overline{a}_{i,j}-\overline{\lambda}'_{k}\overline{a}_{k,j}\right)^{2}\\
 & \le\sum_{j}\left(\sum_{i\in P}\left(\overline{\lambda}_{i}-\overline{\lambda}'_{i}\right)\left|\overline{a}_{i,j}\right|+\overline{\lambda}'_{k}\left|\overline{a}_{k,j}\right|\right)^{2}\\
 & \le\sum_{j}\left(\frac{1}{m}\sum_{i\in P}\overline{\lambda}_{i}\left|\overline{a}_{i,j}\right|+\frac{1}{m}\left|\overline{a}_{k,j}\right|\right)^{2}\\
 & \le\frac{2}{m^{2}}\sum_{j}\underbrace{\left(\sum_{i=1}^{n}\overline{\lambda}_{i}\left|\overline{a}_{i,j}\right|\right)^{2}}_{\text{Jensen inequality: }\sum_{i}\overline{\lambda}_{i}=1}+\frac{2}{m^{2}}\underbrace{\sum_{j}\left|\overline{a}_{k,j}\right|^{2}}_{=1}\\
 & \le\frac{2}{m^{2}}\sum_{j}\sum_{i}\overline{\lambda}_{i}\left|\overline{a}_{i,j}\right|^{2}+\frac{2}{m^{2}}=\frac{2}{m^{2}}\sum_{i}\overline{\lambda}_{i}\underbrace{\sum_{j}\left|\overline{a}_{i,j}\right|^{2}}_{=1}+\frac{2}{m^{2}}=\frac{4}{m^{2}}.
\end{align*}
Similarly, consider the case where $V$ and $P$ have one extra constraint
$a_{k}$ compared with the neighbor $V'$. For all $t\in[T_{1}]$,
for $i\neq k$ we have
\begin{align*}
0\le\lambda_{i}'{}^{(t)}-\lambda_{i}^{(t)} & \le\frac{1}{m^{(t)}-1}-\frac{1}{m^{(t)}}=\frac{1}{m^{(t)}\left(m^{(t)}-1\right)}\le\frac{\lambda_{i}'{}^{(t)}}{m^{(t)}}.
\end{align*}
Hence $\left|\overline{\lambda}_{i}-\overline{\lambda}_{i}'\right|\le\frac{\overline{\lambda}'_{i}}{m}$.
Besides, we also have $\overline{\lambda}_{k}\le\frac{1}{m}$. For
coordinate $j$ 
\begin{align*}
c_{j}-c'_{j} & =\sum_{i\in P'}\overline{\lambda}_{i}\overline{a}_{i,j}-\sum_{i\in P'}\overline{\lambda}'_{i}\overline{a}_{i,j}+\overline{\lambda}_{k}\overline{a}_{k,j}=\sum_{i\in P'}\left(\overline{\lambda}_{i}-\overline{\lambda}'_{i}\right)\overline{a}_{i,j}+\overline{\lambda}_{k}\overline{a}_{k,j}.
\end{align*}
Then
\begin{align*}
\sum_{j}\left(c_{j}-c'_{j}\right)^{2} & =\sum_{j}\left(\sum_{i\in P'}\underbrace{\left(\overline{\lambda}'_{i}-\overline{\lambda}{}_{i}\right)}_{\ge0}\overline{a}_{i,j}+\overline{\lambda}_{k}\overline{a}_{k,j}\right)^{2}\\
 & \le\sum_{j}\left(\sum_{i\in P'}\left(\overline{\lambda}'_{i}-\overline{\lambda}{}_{i}\right)\left|\overline{a}_{i,j}\right|+\overline{\lambda}_{k}\left|\overline{a}_{k,j}\right|\right)^{2}\\
 & \le\sum_{j}\left(\frac{1}{m}\sum_{i\in P'}\overline{\lambda}'_{i}\left|\overline{a}_{i,j}\right|+\frac{1}{m}\left|\overline{a}_{k,j}\right|\right)^{2}\\
 & \le\frac{2}{m^{2}}\sum_{j}\underbrace{\left(\sum_{i=1}^{n}\overline{\lambda}'_{i}\left|\overline{a}_{i,j}\right|\right)^{2}}_{\text{Jensen inequality: }\sum_{i}\overline{\lambda}'_{i}=1}+\frac{2}{m^{2}}\underbrace{\sum_{j}\left|\overline{a}_{k,j}\right|^{2}}_{=1}\\
 & \le\frac{2}{m^{2}}\sum_{j}\sum_{i}\overline{\lambda}'_{i}\left|\overline{a}_{i,j}\right|^{2}+\frac{2}{m^{2}}=\frac{2}{m^{2}}\sum_{i}\overline{\lambda}'_{i}\underbrace{\sum_{j}\left|\overline{a}_{i,j}\right|^{2}}_{=1}+\frac{2}{m^{2}}=\frac{4}{m^{2}}.
\end{align*}
\end{proof}

\section{Missing Proofs from Section \ref{sec:Zero-Margin}}

\begin{proof}[Proof of Lemma \ref{lem:margin-bound}]
Recall that the equality constraint $Cx=g$ is obtained from sanitizing
a set of equality constraints with the entries bounded by $U$ in
the absolute value. In particular, there is a set of tight constraints
$A'x=b'$ in the original LP with entries bounded by $U$ in the absolute
value that spans the same subspace as $Cx=g$---that is, there is
an invertible matrix $M\in\R^{k\times k}$ such that $[C\mid g]=M[A'\mid b']$.

Let $C_{K}$ be $k$ independent columns of $C$, where we let $K$
be the set of indices of the columns we selected; the set $K$ corresponds
to the set of variables that we will eliminate.  Eliminating the
variables in $K$ in $Ax\le b$ using $Cx=g$ means switching to the
new linear constraints $\left(A-A_{K}C_{K}^{-1}C\right)x\le b-A_{K}C_{K}^{-1}g$.
Eliminating the same variables using $A'x=b'$ would get the result
$\left(A-A_{K}A'{}_{K}^{-1}A'\right)x\le b-A_{K}A'{}_{K}^{-1}b'$.
Notice that $C_{K}^{-1}C=\left(MA'_{K}\right)^{-1}MA'=A'{}_{K}^{-1}A'$
and $C_{K}^{-1}g=\left(MA'_{K}\right)^{-1}g=A'{}_{K}^{-1}b'$ so the
two results are identical. 

Therefore, we can think of using $A'x=b'$ for variable elimination
instead of using $Cx=g$. Note that the entries of $A,b,A',b'$ are
bounded by $U$ and the entries of $\left(A'_{K}\right)^{-1}$ are
fractions with denominator $\det\left(A'_{K}\right)$, so to maintain
the integrality during the elimination, we can multiply each row of
$A$ with $\det(A'_{K})$. Therefore, the resulting LP has entries
bounded by $k^{2}(k-1)!U^{k+1}$.

We now show the second claim. First, let us show that the LP $\{Ax\le b;Cx=g,x\ge0\}=\{Ax\le b;A'x=b',x\ge0\}$
has a solution $x$ such that $x_{i}\le((d+1)U)^{d+1}$ for all $i$.
Let $x$ be any vertex solution. By Cramer's rule, $x_{i}$ is the
ratio of two determinants with integer entries $\le U$. We can bound
the numerator of this ratio by $((d+1)U)^{d+1}$ and lower bound the
denominator by $1$.

Next, consider the LP with added slack $\left\{ Ax\leq b+\eta\one,A'x=b',x\ge0\right\} $.
After performing the variable elimination as described above, we obtain
an LP $\left\{ \tilde{A}x\leq\tilde{b}+\eta\one,x\ge0\right\} $.
Finally, we bound the margin of the vertex solution $x$ (restricted
to the variables that were not eliminated) for the homogenized version
of the latter LP . Since $\tilde{A},\tilde{b}$ have entries bounded
by $k^{2}(k-1)!U^{k+1}$ and $x$ has entries bounded by $((d+1)U)^{d+1}$,
the margin of the homogenized LP is at least 
\begin{align*}
\frac{\eta}{\left\Vert (x\mid1)\right\Vert \cdot\max_{i}\left\Vert [-\tilde{A}_{i}\mid\tilde{b}_{i}]\right\Vert }\ge\frac{\eta}{(d+1)\left(k^{2}(k-1)!U^{k+1}\right)\cdot((d+1)U)^{d+1}} & \ge\frac{\eta}{((d+1)U)^{2(d+1)}}.
\end{align*}
\end{proof}

To show Lemma \ref{lem:general-utility}, we restate the following
lemma from \cite{kaplan2024differentially} shows that tightness in
the relaxed system implies tightness in the original system.
\begin{lem}[Lemma 36 \cite{kaplan2024differentially}]
\label{lem:tight-constraint}For matrices $A_{1},A_{2}$ with dimension
$m_{1}\times d$, $m_{2}\times d$ and $b_{1},b_{2}$ being vectors
of length $m_{1}$, $m_{2}$. The entries are integers with upper
bound $U$ in the absolute value. For $\eta_{2}\ge0$ being a vector
of length $m_{2}$ such that the entries of $\eta_{2}$ are bounded
by $\frac{1}{2(d+1)((d+1)U)^{d+1}}$. Then if the system 
\begin{align*}
A_{1}x & \le b_{1}\\
A_{2}x & =b_{2}+\eta_{2}\\
x & \ge0
\end{align*}
is feasible then the system 
\begin{align*}
A_{1}x & \le b_{1}\\
A_{2}x & =b_{2}\\
x & \ge0
\end{align*}
is feasible.
\end{lem}
\begin{proof}[Proof of Lemma \ref{lem:general-utility}]
In each iteration, the algorithm solves the LP $\tilde{A}x\le\tilde{b}+\eta\one$
with a set of equality constraints $Cx=g$, where $\tilde{A}x\le\tilde{b}$
is a subset of the constraints of the original LP $Ax\le b$. Note
that, by construction, $Cx=g$ is equivalent to a set of tight constraints
$A'x=b'$ in the original LP. Solving this LP gives us a solution
that satisfies 
\begin{align*}
A_{1}x & \le b_{1}\\
A_{2}x & =b_{2}+\eta_{2}\\
A'x & =b'
\end{align*}
for a vector $\eta_{2}\le\eta\one$, where $A_{1}x\le b_{1}$ and
$A_{2}x\le b_{2}$ are constraints in the original LP. Since all
entries of the above LP are bounded by $U$ and $\eta=\frac{1}{2(d+1)((d+1)U)^{d+1}}$,
Lemma \ref{lem:tight-constraint} implies that the LP 
\begin{align*}
A_{1}x & \le b_{1}\\
A_{2}x & \le b_{2}\\
A'x & =b'\qquad\text{equivalently, }Cx=g
\end{align*}
is feasible. The $\PP$ algorithm succeeds with probability $1-\beta$
and the sanitization step succeeds with probability $1-\delta$. Therefore,
after each iteration, with probability at least $1-\left(\beta+\delta\right)$
the set of tight (equality) constraints is increased by at least $1$.
Hence, with probability at least $1-d(\beta+\delta)$ the algorithm
must terminate after at most $d$ iterations.

By Lemma \ref{lem:margin-bound}, in the homogenized LP after variable
elimination, the margin is at least $\frac{\eta}{((d+1)U)^{2(d+1)}}\ge\eta^{3}=\rho$.
Therefore, with probability $\ge1-\beta$, the number of dropped constraints
when using $\PP$ in Line \ref{drop-perceptron} is at most 
\begin{align*}
O\left(\frac{d^{2}}{\epsilon}\log^{2}\frac{d}{\beta\delta}\sqrt{\log\frac{1}{\rho}}\right) & =O\left(\frac{d^{2.5}}{\epsilon}\log^{2}\frac{d}{\beta\delta}\sqrt{\log(dU)}\right).
\end{align*}
If in iteration $t$, the algorithm returns at Line \ref{return-1},
the number of dropped constraints in $J_{2}$ is at most $\frac{d^{2}}{\epsilon}$
with probability at least
\begin{align*}
1-\Pr\left[\Lap\left(\frac{1}{\epsilon}\right)<-\log\frac{1}{\delta}\right] & \ge1-\delta.
\end{align*}
If the algorithm returns at Line \ref{return-2}, again, the number
of dropped constraints is at most $O\left(\frac{d^{2}}{\epsilon}\log^{2}\frac{d}{\beta\delta}\sqrt{\log\frac{1}{\rho}}\right)$
with probability $\ge1-\delta$. By Theorem \ref{thm:sanitize}, the
number of dropped constraints by sanitization in each iteration is
at most $O\left(\frac{d^{2}}{\epsilon}\log\frac{d}{\delta}\right)$
with probability $\ge1-\delta$. 

Combining these, over all iterations, with probability at least $1-d(\beta+\delta)$,
the number of dropped constraints is at most $O\left(\frac{d^{3.5}}{\epsilon}\log^{2}\frac{d}{\beta\delta}\sqrt{\log(dU)}\right).$
\end{proof}

\end{document}